\documentclass[aps,prl,twocolumn,bibliography,superscriptaddress]{revtex4-1}
\usepackage[utf8]{inputenc}
\usepackage{amsmath}
\usepackage{mathtools}
\usepackage{amsfonts}
\usepackage{amssymb}
\usepackage{graphicx}
\usepackage{amsthm}
\usepackage{textcomp}
\usepackage{color}
\usepackage[dvipsnames]{xcolor}

 \usepackage{combelow}

\newtheorem{theorem}{Theorem}
\numberwithin{theorem}{section}
\newtheorem{proposition}[theorem]{Proposition}

\newtheorem{definition}[theorem]{Definition}

\begin{document}

\title{A Tropical Geometric Approach To Exceptional Points}

\author{Ayan Banerjee}
\affiliation{Solid State and Structural Chemistry Unit, Indian Institute of Science, Bengaluru 560012, India}
\author{Rimika Jaiswal}
\affiliation{Undergraduate Programme, Indian Institute of Science, Bengaluru 560012, India}
\affiliation{Department of Physics, University of California, Santa Barbara, California 93106-9530, USA}
\author{Madhusudan Manjunath}
\affiliation{Department of Mathematics, Indian Institute of Technology Bombay, Powai, Mumbai 400076, India}
\author{Awadhesh Narayan}
\email{awadhesh@iisc.ac.in}
\affiliation{Solid State and Structural Chemistry Unit, Indian Institute of Science, Bengaluru 560012, India}

\date{\today}

\begin{abstract}
Non-Hermitian systems have been widely explored in platforms ranging from photonics to electric circuits. A defining feature of non-Hermitian systems is exceptional points (EPs), where both eigenvalues and eigenvectors coalesce. Tropical geometry is an emerging field of mathematics at the interface between algebraic geometry and polyhedral geometry, with diverse applications to science. Here, we introduce and develop a unified tropical geometric framework to characterize different facets of non-Hermitian systems. We illustrate the versatility of our approach using several examples, and demonstrate that it can be used to select from a spectrum of higher-order EPs in gain and loss models, predict the skin effect in the non-Hermitian Su-Schrieffer-Heeger model, and extract universal properties in the presence of disorder in the Hatano-Nelson model. Our work puts forth a new framework for studying non-Hermitian physics and unveils a novel connection of tropical geometry to this field.  
\end{abstract}
\maketitle

\section{Introduction}
\setlength{\parskip}{0.5\baselineskip}

Several branches of mathematics show an \emph{unreasonable effectiveness} in formulating and understanding a myriad of physical phenomena~\cite{wigner1990unreasonable}. Striking recent examples include the role of topology in condensed matter systems~\cite{haldane2017nobel,kosterlitz2017nobel}, advent of knot theory in quantum field theory~\cite{witten1989quantum}, and applications of graph theory in statistical mechanics~\cite{essam1971graph}.

Tropical geometry is a branch of modern mathematics at the interface between algebraic geometry and polyhedral geometry~\cite{maclagan2009introduction,mikhalkin2009tropical}. The tropical approach has not only had applications to geometry, but also to areas such as physics, number theory, genetics, economics, optimization theory, and computational biology~\cite{brugalle2014bit,kapranov2011thermodynamics,samal2015geometric,noel2012tropical}. Notable has been the role of tropical geometry in understanding physical systems. Deep connections of tropical geometry to string theory have been discovered~\cite{kontsevich2001homological,gross2011tropical}, while tropical algebra has been used to analyze frustrated systems such as spin ice and spin glasses~\cite{cimasoni2007dimers}. Another recent successful application of tropical ideas has been in understanding self-organized criticality in dynamical systems~\cite{kalinin2018self}. Tropical geometric tools such as the logarithmic transformation offer drastic computational simplification, and, interestingly, the low-temperature limit of statistical physics can be studied in terms of such a tropical mapping~\cite{kenyon2006dimers,kapranov2011thermodynamics}.

Hermiticity of operators is a central principle in quantum mechanics, ensuring that a system has real eigenenergies and orthogonal eigenstates, and leads to the conservation of probability~\cite{dirac1981principles}. In recent decades the notion of non-Hermiticity has been introduced in a variety of physical contexts~\cite{bender1998real,bergholtz2021exceptional,ashida2020non}. A unique feature of non-Hermitian systems are degeneracies called \emph{exceptional points} (EPs), where both eigenvalues and eigenvectors coalesce~\cite{kato2013perturbation}. The energy level-splitting, $\Delta\lambda$, upon moving away from an EP follows a distinctive fractional dependence on the perturbation. An $N$-th order EP [EP-$N$, where two or more eigenvectors $(N\geq 2)$ coalesce] shows a splitting of the form $\Delta \lambda \sim \nu^{1/N}$, where $\nu$ is an external perturbation~\cite{hodaei2017enhanced,tang2020exceptional}. Recent advances have led to controllable realization of EPs in a variety of platforms~\cite{miri2019exceptional,ozdemir2019parity,naghiloo2019quantum,cerjan2019experimental,xu2017weyl,choi2018observation,stehmann2004observation}. Their control has enabled the exploration of novel phenomena, such as uni-directional sensitivity~\cite{lin2011unidirectional,peng2014parity}, laser mode selectivity~\cite{feng2014single,hodaei2014parity}, and non-Hermitian skin effect (NHSE)~\cite{yao2018edge}.

In this work, we propose and develop a general tropical geometric framework for understanding and characterizing various facets of non-Hermitian systems. We demonstrate that the tropical geometric information encoded in the characteristic polynomial of the non-Hermitian Hamiltonian can be used to identify and classify EPs using \emph{valuation} and \emph{tropical roots} -- concepts that naturally emerge in the tropical setting. We show that EPs of different orders and their transitions can be captured in an elegant manner by \emph{amoebas} and Newton polygons. We illustrate our framework using experimentally-realized gain and loss models, and show how it allows obtaining a higher-order EP or choosing from a spectrum of EPs. Using the paradigmatic non-Hermitian Su-Schrieffer-Heeger (SSH) model, we demonstrate how our tropical geometric approach can be used to predict the NHSE. Our approach naturally allows extracting the universal properties of EPs in the presence of disorder, which we highlight using the celebrated Hatano-Nelson model. Our framework allows a unified approach to different facets of non-Hermitian phenomena, including EPs, NHSE, and holonomy.

\begin{figure}	
	\includegraphics[width=0.83\linewidth]{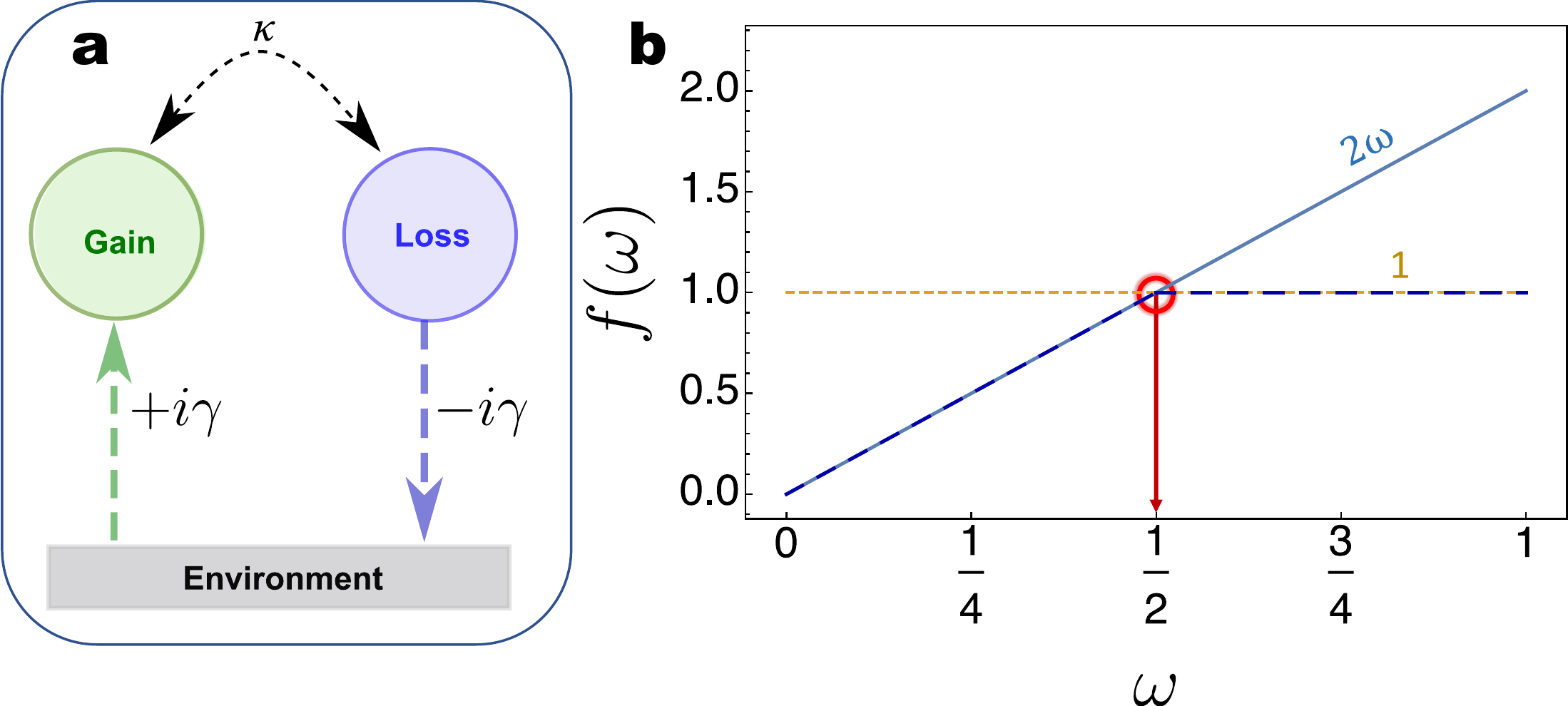}
	\caption{\textbf{Illustration of the tropical geometric framework with two site gain and loss model.} (a) Schematic of the two site gain and loss model with $\gamma$ as the gain and loss parameter and $\kappa$ as the coupling between the sites. (b) Using tropicalization to find the order of EP. The tropical polynomial contains different linear monomials with integer coefficients (see Eq.~\ref{trop-two}). The bend locus of the monomials (encircled in red) gives the tropical root $\omega_0=1/2$ implying a second order EP.}
	\label{two-site}	
 \end{figure}

\subsection{Tropical characterization of exceptional points}
 
\textbf{Basics of tropical geometry.}  We begin by briefly summarizing the fundamental ideas of tropical geometry (see supplemental material~\cite{supplement} for a detailed discussion). Broadly speaking, tropical geometry studies solutions of systems of polynomials by transforming them into piece-wise linear subsets of Euclidean space~\cite{maclagan2021introduction}. The basic algebraic object underlying tropical geometry is the \emph{tropical semiring}, $(\mathbb{R}\cup \{\infty\},\oplus,\odot)$. This denotes a set that is the union of the set of real numbers $\mathbb{R}$, together with an element ``infinity'', and two operations on it, namely tropical addition $\oplus$ and tropical multiplication $\odot$. The tropical sum of two numbers is their minimum and the tropical product is their usual sum,

\begin{equation}
    x \oplus y= \mathrm{min}(x,y), \quad  x \odot y = x+y.
\end{equation}

Many of the usual axioms of arithmetic remain valid in the tropical setting. These operations satisfy all the ring axioms except for the existence of an additive inverse and thus turn $(\mathbb{R}\cup \{\infty\},\oplus,\odot)$ into a semiring.

\textbf{Defining order of exceptional points.} In the following, we define the notion of order of an EP of a non-Hermitian system. To the best of our knowledge, this definition is consistent with the literature on this topic. Let $H(\nu)$ be the Hamiltonian of a non-Hermitian system in one variable $\nu$ with an EP at $\nu=0$ and let $p(\nu,\lambda) \in \mathbb{C}[\nu, \lambda]$ be its characteristic polynomial.  
In the following, we regard $p$ as a polynomial in one variable $\lambda$ and with coefficients in the field $\mathbb{C}\{\{\nu\}\}$ of Puiseux series. 

\begin{definition}
Let $p \in \mathbb{C}\{\{\nu\}\}[\lambda]$ have at least one non-zero root.  Suppose that $p$ has a non-trivial Puiseux series root, i.e. a root  $s$ such that  the least exponent of $s$ is non-zero.  In this case, the order of this EP (at $\nu=0$) is the maximum absolute value of the denominator $n$ of $m/n \in \mathbb{Q}$ (in reduced form) where $m/n$ varies over the least exponent in the Puiseux series expansion over all the non-trivial roots of $p$. Otherwise, if all the roots of $p$ have zero as their least exponent, then $\nu=0$ is called a degenerate point. 
\end{definition}

Consider a system at an EP-$N$ (at $\nu=0$) given by the Hamiltonian $H_0(x_1, x_2, ...)$ where $x_1, x_2...$ are system-dependent parameters. When we perturb this system around the EP, the eigenvalues of the perturbed Hamiltonian $H(\nu)=H_0+\nu H_1$ follow a Puiseux series in $\nu$,

\begin{equation}
    \lambda(\nu)= \gamma_1 \nu^{1/N} + \gamma_2 \nu^{2/N}+...,
    \label{puiseux}
\end{equation}

where $\nu$ is the perturbation strength. To leading order, the response goes as $\Delta\lambda_{EP-N} \propto \nu^{1/N}$. Our tropical geometric approach features a characterization of EPs by determining such leading order behavior.

\textbf{Characterizing exceptional points using tropical geometry.} Next, we present the tropical geometric framework that can be used to reveal the structure of EPs and characterize as well as tune them in various physical platforms. 

For a field $\mathbb{K}$, a valuation on $\mathbb{K}$ is defined as a function $\text{val} :\mathbb{K}\rightarrow \mathbb{R}  \cup \{\infty\}$ such that:
    \begin{itemize}
        \item $\text{val}(a) = \infty$ if and only if $a = 0$;
        \item $\text{val} ( ab ) = \text{val} ( a ) + \text{val} ( b )$ ;
        \item $\text{val} ( a + b ) \geq \text{min} \{ \text{val} ( a ) , \text{val} ( b )\}$ for all $a, b \in \mathbb{K}$.
    \end{itemize}
In our framework, we primarily deal with the field of Puiseux series with coefficients in the complex numbers $\mathbb{C}$. This field has a natural valuation which is given by taking a non-zero Puiseux series to the lowest exponent that appears in its expansion. For example, $\text{val}(t^2-2t+3)=\text{min}\{ \text{val}(t^2), \text{val}(-2t), \text{val}(3) \} = \text{min} \{ 2,1,0 \} = 0$ and $\text{val}(t^{1/2}-t^{3/4}+t^1+t^2+\dots)=1/2$.

In its most basic form, tropical geometry gives a method to compute the valuations of the non-zero roots of a non-zero polynomial $p \in \mathbb{K}[\lambda]$ in terms of the valuations of the coefficients of $p$. More precisely, given a non-zero polynomial $p=\sum_{i=0}^{d} a_i \lambda^i \in \mathbb{K}[\lambda]$, its tropicalization ${\rm trop}(p): \mathbb{R} \rightarrow \mathbb{R}$ is defined as ${\rm trop}(p)(\omega)={\rm min_i}\{{\rm val}(a_i)+i \cdot \omega \}$. 

A real number $\omega_0$ is called a \emph{tropical root} of ${\rm trop}(p)$ if the minimum defining ${\rm trop}(p)(\omega_0)$ is attained by at least two distinct terms ${\rm val}(a_j)+j \cdot \omega_0$ and ${\rm val}(a_k)+k \cdot \omega_0$ for $j \neq k$. Equivalently, the tropical roots of ${\rm trop}(p)$ precisely are the real numbers where ${\rm trop}(p)$ is not differentiable, called the \emph{bend locus} of ${\rm trop}(p)$.

The \emph{fundamental theorem of tropical geometry} asserts that the set of tropical roots of ${\rm trop}(p)$ is precisely the set of valuations of the non-zero roots of $p$~\cite[Chapter 3, Section 2]{maclagan2021introduction}. This leads us to one of the main propositions of our framework. For a non-Hermitian Hamiltonian, $H(\nu)$, with a characteristic polynomial $p(\nu,\lambda) \in \mathbb{C}[\nu, \lambda]$, as described before, $p$ can be regarded as an element in $\mathbb{C}\{\{\nu\}\}[\lambda]$, where $\mathbb{C}\{\{\nu\}\}$ is equipped with its standard valuation that takes a non-zero Puiseux series $s$ to the exponent of the leading order term of $s$.

\begin{proposition}
Suppose that ${\rm trop}(p(\nu,\lambda))$ has a non-zero tropical root. The order of the EP at $\nu=0$ of $H(\nu)$ is the maximum absolute value of the denominator $n$ of $m/n$ (in reduced form) where $m/n$ varies over all the non-zero tropical roots of ${\rm trop}(p(\nu,\lambda))$. 
Otherwise, if ${\rm trop}(p(\nu,\lambda))$ has no non-zero tropical roots, then $\nu=0$
is a degenerate point.  
\end{proposition}

\begin{proof}
By the fundamental theorem of tropical geometry \cite[Chapter 3, Section 2]{maclagan2021introduction}, the set of tropical roots of ${\rm trop}(p)$ is precisely the set $\{{\rm val}(s)\}_s$  where $s$ varies over all the non-zero Puiseux series solutions of $p(\nu,\lambda) \in \mathbb{C}\{\{\nu\}\}[\lambda]$. With this information at hand, the statement follows from the definition of the order of an EP. 
\end{proof}

To simply illustrate our framework, we consider an experimentally realizable non-Hermitian system consisting of two coupled sites with gain and loss (see Fig.~\ref{two-site}a). The Hamiltonian reads

\begin{equation}
H_2=\begin{pmatrix}
 \alpha + i \gamma & \kappa\\ 
  \kappa & -\alpha-i \gamma
\end{pmatrix}.
\end{equation}

Here $\alpha$ quantifies the onsite energies, $\gamma$ is the corresponding gain/loss coefficient, and $\kappa$ is the coupling between the sites. This system has an EP at $\alpha=0$ if $\gamma=\kappa$. The characteristic polynomial and the corresponding tropicalization for $\gamma=\kappa$ are

\begin{align}
    & p(\alpha,\lambda)= -2i\kappa \alpha - \alpha^{2} + \lambda^{2}, \\
    & \text{trop}\left( p(\alpha,\lambda) \right)(\omega) = \text{min} \left( 1, 2\omega \right).
    \label{trop-two}
\end{align}

The root of the tropical polynomial is given by the bend locus of $\text{trop}(p(\alpha,\lambda))(\omega)$ which occurs at $\omega_0=1/2$, as shown in Fig.~\ref{two-site}b. Using the fundamental theorem of tropical geometry, we then conclude that $p(\alpha,\lambda)$ has a non-zero root with valuation $s=1/2$. This implies that the roots of $p(\alpha,\lambda)$, i.e., the eigenvalues of $H$ have the form $\lambda \sim \alpha^{1/2}$ near the EP at $\alpha=0$. Thus, the EP at $\alpha=0$ is a second-order EP (see the supplement~\cite{supplement} for a detailed discussion). Further, in the supplement~\cite{supplement}, we use tropicalization to illustrate how our framework provides a natural way to characterize and tune to higher-order EPs using companion matrices.

\begin{figure}	
	\includegraphics[width=0.6\linewidth]{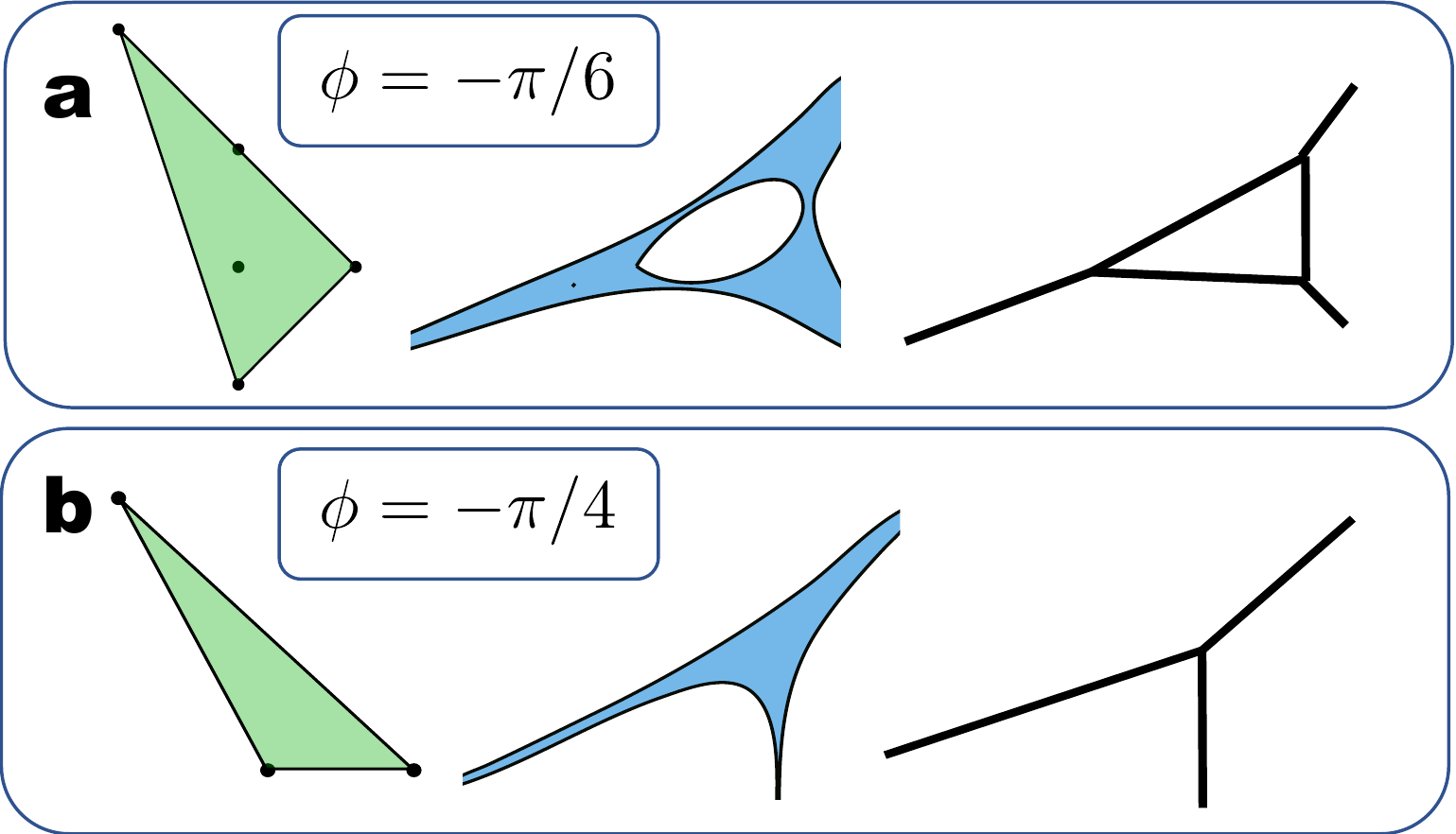}
	\caption{\textbf{Characterization of exceptional points through amoebas in a three-site gain and loss model.} Realization of Newton polygon (left), amoeba (center), and the spine of the amoeba for (a) $\phi=-\pi/6$ and (b) $\phi=-\pi/4$. The convex slope of the Newton polygon defines the order of EPs. We obtain third-order EPs in (a) and second-order EPs in (b). The interior point in the Newton polygon in (a) results in a vacuole in the amoeba. The amoeba structures abruptly change from (a) to (b) while transitioning from third-order to second-order EPs. We set $\gamma = \sqrt{2} \kappa$ and $\kappa=1.0$. }
	 \label{three-site-newton-amoeba}	
\end{figure}

\textbf{Relation to amoebas and Newton polygons.} Above, we saw how tropical geometry can be used to determine the order of EPs. A precursor to tropical geometry is a construction called the \emph{amoeba of a complex algebraic variety} due to Gelfand, Kapranov and Zelevinsky \cite{gelfand1994discriminants}. Let $V \subseteq (\mathbb{C}^\star)^n$ be the set of solutions, all of whose coordinates are non-zero, of a finite set of Laurent polynomials in $n$ variables. Let ${\rm Log}: (\mathbb{C}^\star)^n \rightarrow \mathbb{R}^n$ be the logarithmic map that takes $(z_1,\dots,z_n)$ to $(\log (|z_1|),\dots,\log(|z_n|))$. The amoeba of $V$ is the image of the logarithmic map restricted to $V$. A related and important notion is the \emph{spine of the amoeba} and is defined as the limit as $t \rightarrow \infty$ of the parameterized logarithmic map ${\rm Log}_t(z_1,\dots,z_n)=(\log_t (|z_1|),\dots,\log_t(|z_n|))$. In the Methods section we present the connection of amoeba to tropicalization.

\begin{figure}	
	\includegraphics[width=0.85\linewidth]{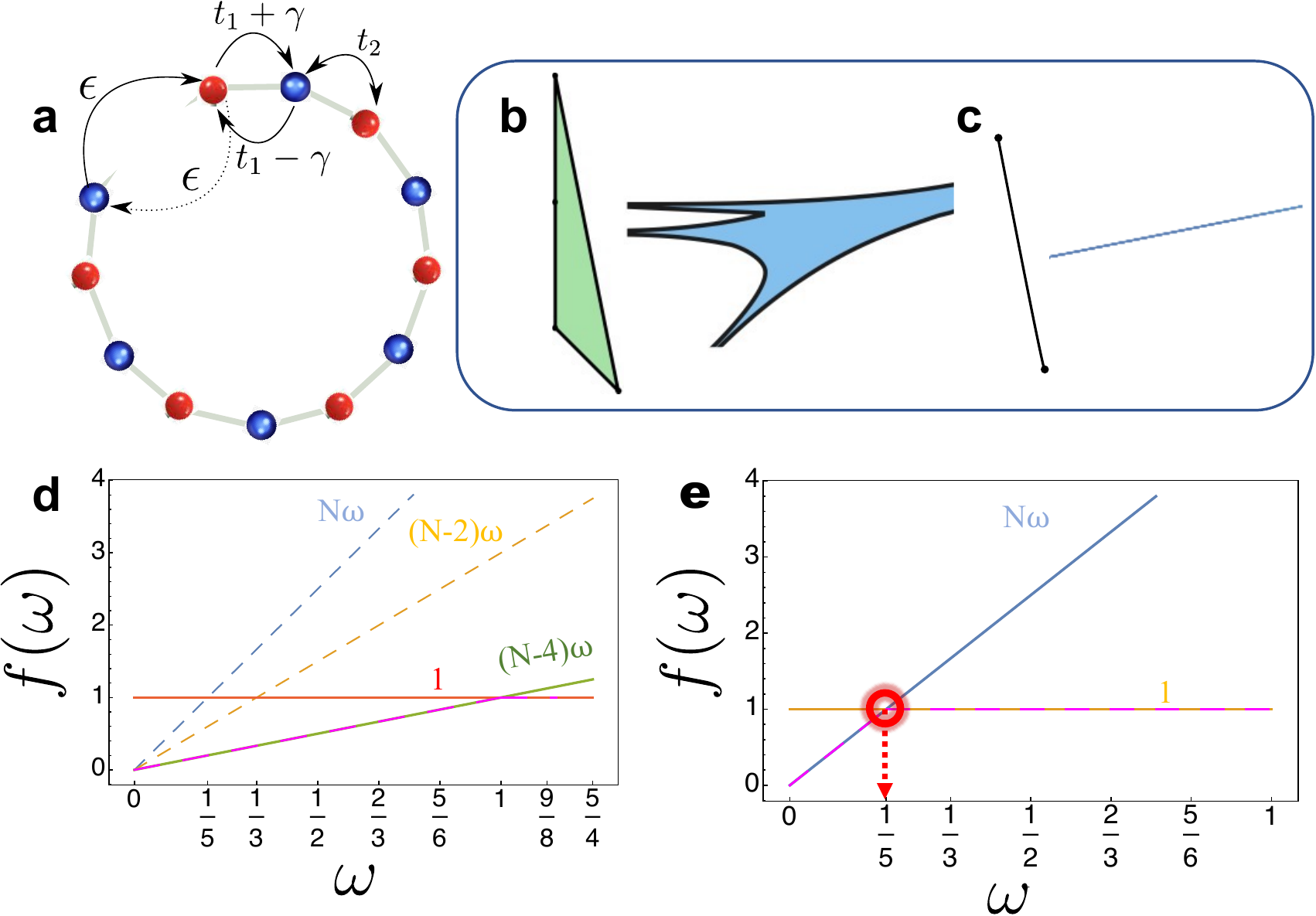}
	\caption{\textbf{Detecting skin effect in the non-Hermitian SSH model with higher-order EPs via tropical geometry} (a) Schematic of SSH model with non-reciprocal hopping and a weak link connecting the last site to the first. The inter-unit cell hopping is a constant $t_2$, but the left and right intra-unit cell hoppings are given by $t \pm \gamma$ incorporating non-Hermiticity in the system. (b) Newton polygons and the concomitant amoebas for SSH model with odd number of sites. Here we choose $t_1=2.0$, $\gamma=1.0$ and $t_2=1.0$. The structure is similar for all values $t_1 \neq \gamma$.  (c) At $t_1=\gamma$ the Newton polygon abruptly transforms to a single line with slope $1/N$ at $t_1=\gamma$, indicative of a higher order EP and the skin effect. The amoeba collapses to a single line perpendicular to the Newton polygon, characterizing the non-Hermitian phase transition. Panel (d) shows the tropicalization for the general case corresponding to Eq.~\ref{ssh-tropicalization}. Each straight line represents a term in Eq.~\ref{ssh-tropicalization}. (e) Tropicalization for the case of $t_1=\gamma$, wherein the coefficients of all the points corresponding to $\lambda^M$ vanish, other than $M = 0, N$. The bend locus gives the tropical root $\omega_0=1/N$, which indicates the presence of an $N$-th order EP and, correspondingly, the occurrence of a non-Hermitian skin effect. Here we choose $N=5$.}
	\label{ssh-figure}	
\end{figure}

We are primarily concerned with amoebas of polynomials in two variables with complex coefficients, namely characteristic polynomials of a non-Hermitian Hamilotian in one variable. Typical examples of such amoebas are shown in Fig.~\ref{three-site-newton-amoeba}, which will be discussed shortly. The amoeba of a typical polynomial contains (unbounded) rays that are called its \emph{tentacles}. We recall that the Newton polygon of $p$ is the convex hull of the exponents of the monomials in the support of $p$. The following proposition that relates the edges of the Newton polygon of $p$ and the amoeba (of the algebraic variety) associated to $p$ is of fundamental importance to our framework. 

\begin{proposition}
The set of directions of the tentacles of the amoeba associated to $p$ is precisely the set of outer normals of the edges of the Newton polygon of $p$.
\end{proposition}
We refer to Proposition 1.9 ~\cite{gelfand1994discriminants} and Section 1.4~\cite{maclagan2021introduction} for a more general version of this proposition.

We illustrate this proposition using a three-site non-Hermitian trimer model with balanced gain and loss and an asymmetric onsite potential. The Hamiltonian for the trimer is

\begin{equation}
H_3=\begin{pmatrix}
 \alpha + i \gamma & \kappa & 0\\ 
 \kappa & 0 & \kappa\\
 0& \kappa & \beta-i \gamma
\end{pmatrix},
\end{equation}

where the different symbols have a meaning analogous to the two-site model. We use the transformation $\beta= \alpha \tan{\phi}$ to scan all angles in the $\alpha$-$\beta$ parameter plane. Using the formalism developed above, we can find the tropical roots to reveal the nature of EPs. In Fig.~\ref{three-site-newton-amoeba} we illustrate the amoebas and the concomitant Newton polygons for various $\phi$. Note that the steepest slope of the Newton polygon $\Delta$ determines the order of the EPs. Interestingly, the integer points of the Newton polygon correspond to the vacuole in the amoeba (see Fig.~\ref{three-site-newton-amoeba}a). The structure of the amoeba drastically transforms from $\phi=-\pi/4$ to $\pi/4$ while the tropical roots change from $1/2$ to $1/3$ with a transition from second-order to third-order EPs. Therefore, the structure of the amoeba can be directly used to identify the various non-Hermitian phases.

\begin{figure*}
 \includegraphics[scale=0.4]{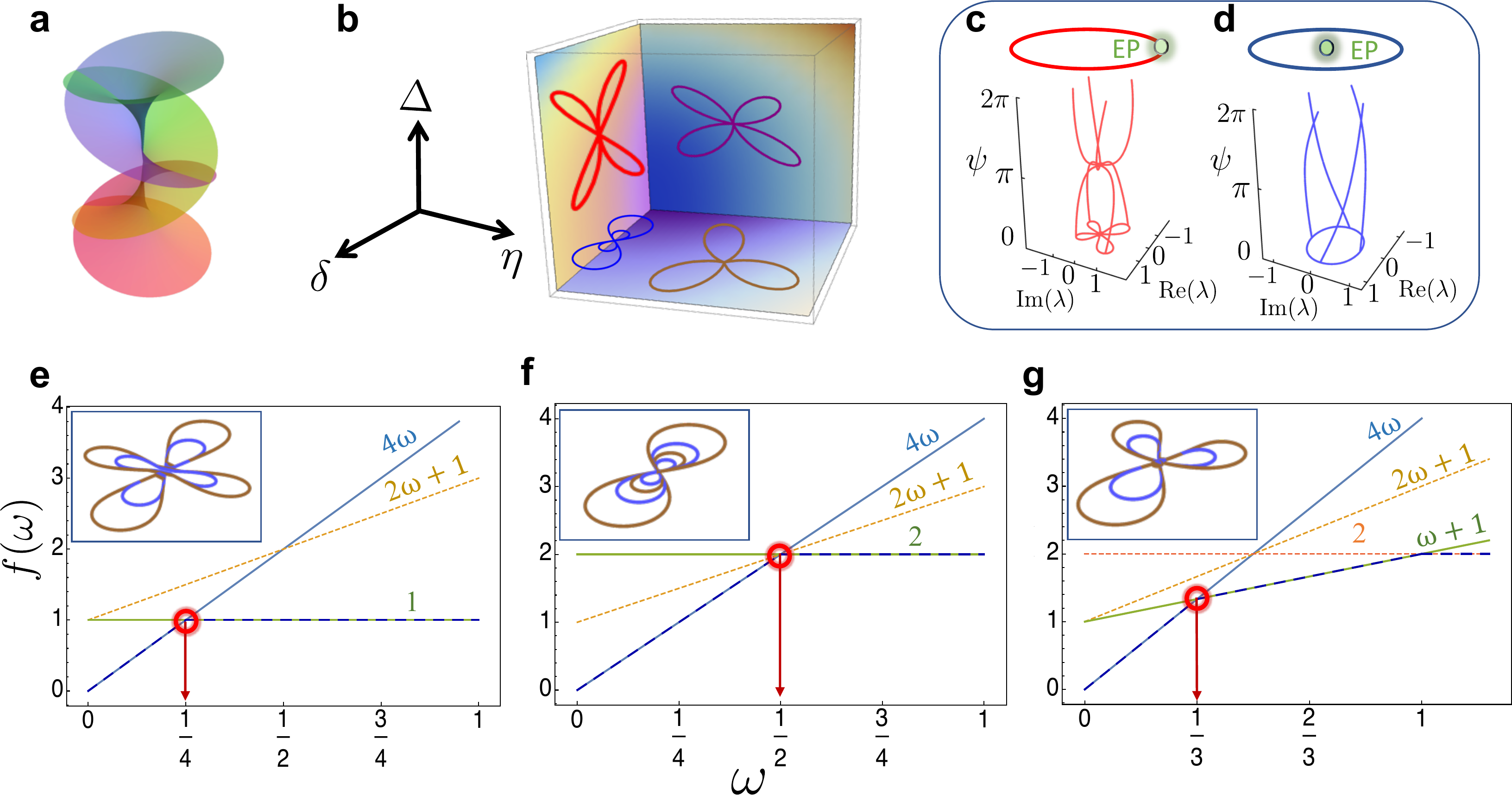}
   \caption{\textbf{Holonomy characterization of disordered Hatano-Nelson model.} (a) Riemann surface for a quartic root in the complex plane. (b) The Hatano-Nelson model exhibits anisotropic exceptional behaviour in the parameter space as illustrated by the different projection of eigenbands along different parameter planes (the number of petals representing the order of EP) (c), (d) Swapping of eigenmodes arising from Riemann sheet exchange while tracing a loop in parameter space given by $R= c e^{i \psi}, \psi \in (0,2 \pi)$. We show the holonomy properties when the loop (c) critically touches EP and (d) encloses the EP. Note that in (d) the $N$ eigenmodes undergo a cyclic permutation among themselves while in (c) eigenmode evolution forms $N$ petals in the complex energy plane where $N$ is the order of EP.   Tropicalization and tropical roots showing (e) fourth- $(\theta=0, \phi=\pi/4)$, (f) second- $(\theta=0, \phi=0)$, and (g) third-order $(\theta=\pi/4,\phi=0)$ EPs for different values of $\phi$ and $\theta$. The insets show holonomy characterization in the presence (brown) and absence (blue) of disorder. Disorder preserves the stability of EPs but renormalizes the spectral properties.} \label{dis-1}
\end{figure*}

\subsection{Tropical analysis of the Su-Schrieffer-Heeger model} 

Below we apply the complete tropical approach developed in this work to the paradigmatic non-Hermitian SSH Hamiltonian with non-reciprocal hopping~\cite{lieu2018topological,herviou2019defining}. We demonstrate how our tropical geometric approach can detect the NHSE, which is a unique feature of non-Hermitian systems where a large number of states accumulate at boundaries of open systems~\cite{yao2018edge,weidemann2020topological}. The Hamiltonian of the non-Hermitian SSH model reads

\begin{align}
H_{SSH} =  - & \sum_{i} [t_1 (c^\dagger_{i,A} c_{i,B} + h.c.) +
t_2 (c^\dagger_{i+1,A} c_{i,B} + h.c.)]  \notag \\ +& \sum_{i}
\gamma (c^\dagger_{i,B}c_{i,A} - c^\dagger_{i,A}c_{i,B}),
\label{eqn:Ham}
\end{align}

where $c^{\dagger}_{i,\alpha} (c_{i,\alpha})$ is the fermionic creation (annihilation) operator at site $i$ for sublattice $\alpha=A,B$. The intra- and inter-unit cell hopping amplitudes are given by $t_1$ and $t_2$, respectively, and $\gamma$ introduces a non-reciprocity only in the intra-unit cell hopping, resulting in non-Hermiticity (see Fig.~\ref{ssh-figure}a). We introduce a perturbation, $\sigma t_2$, $\sigma \in [0,1]$, which connects the last and first sites. The Hamiltonian takes the matrix form ($\epsilon = \sigma t_2$)

\begin{equation}
    H_{SSH}(\epsilon)=
    \begin{bmatrix}
    0 & t-\gamma & 0 & \cdots & \epsilon \\
    t+\gamma & 0 & t_2 & \cdots & 0 \\
    0 & t_2 & 0 & \cdots & 0\\
    \vdots & \vdots & \vdots & \ddots & \vdots \\
    \end{bmatrix}_{N \times N}.
\label{eq:HN_model_pert}
\end{equation}

The characteristic equation, in turn, is 

\begin{equation}
\begin{split}
p(\epsilon,\lambda) & = \text{mod}(N-1,2)\{\gamma^2-t_1^2\}^{\frac{N}{2}} - t_2^{\frac{N-2}{2}}  (t_1+\gamma)^{N/2} \epsilon \\
  & +\sum_{M=N,N-2,\cdots} [z_M\{\gamma^2-t_1^2\}^{\frac{N-M}{2}}+t_2 \mathcal{O}_M]\lambda^{M},
\end{split}
\label{ssh-characteristics}
\end{equation}

where each $z_M$ is a constant, $z_M\in \mathbb{Z}$. The tropical polynomial is calculated to be

\begin{equation}
\text{trop} \left( p(\epsilon,\lambda) \right) (\omega) = \text{min}\{m, \cdots (N-2) \omega, N \omega \},
\label{ssh-tropicalization}
\end{equation}

where $m=$ 0 (1) for even (odd) sites. The tropicalization and bend locus for $p(\epsilon,\lambda)$ are shown in Fig.~\ref{ssh-figure}d and \ref{ssh-figure}e. Strikingly at $t_1=\gamma$ and $t_2\rightarrow 0$, the coefficients of all the terms in $p(\epsilon,\lambda)$ vanish other than the $\epsilon^1$ and the $\lambda^N$ terms which lead to the solution $\lambda = \epsilon^{1/N}$. This fractional exponent, in turn, shows that higher-order EPs appear for $t_1 = \pm \gamma$ with an algebraic multiplicity that scales with system size while the geometric multiplicity remains unity. This is a signature of the NHSE, wherein all the bulk modes collapse to one state and are exponentially localized at the edge under open boundary conditions.

This physics can be beautifully captured by the amoebas and their corresponding Newton polygons. In Fig.~\ref{ssh-figure}b and \ref{ssh-figure}c, we present the Newton polygon and associated amoeba for this model. The edges of the Newton polygon are perpendicular to the tentacles of the amoeba. The structure of the amoeba remains invariant unless we have $t_1=\gamma$, where the amoeba and the corresponding Newton polygon strikingly collapse to a single straight line, as shown in Fig.~\ref{ssh-figure}c. The Newton polygon in the latter case has a slope of $1/N$, establishing the presence of an EP-$N$.

\subsection{Application to disorder and holonomy}

Our tropical geometric framework can be used to extract universal properties of EPs even in the presence of disorder as we show next. To illustrate, we consider the celebrated Hatano-Nelson model~\cite{hatano1996localization} under open boundary conditions with $N$ sites, along with upper corner perturbations, i.e., additional couplings in the $(1,j)$-th entries of the Hamiltonian, where $j=N, N-2 \cdots$. For $4$ sites, the Hamiltonian reads

\begin{equation}
 H_4 =  \begin{pmatrix}
    0 & \delta & \eta & \Delta \\
    (2+\delta) & 0 & \delta & 0 \\
    0 &  (2+\delta) & 0 & \delta \\
    0 & 0 &  (2+\delta) & 0 \\
    \end{pmatrix}.
\end{equation}

The Hamiltonian can be written as a sum of two companion matrices indicating that it features different exceptional behaviour along different sections of the parameter space. To study the structure of EPs in the parameter space of $\delta$, $\Delta$ and $\eta$, we shift to generalized spherical coordinates $\delta= r\cos{\theta}\cos{\phi}$, $\Delta=r\cos{\theta}\sin{\phi}$, $\eta=r\sin{\theta}$ and study the tropicalization of the characteristic equation $p(r,\lambda)$. We find highly anisotropic behaviour resulting in EP-2, EP-3 or EP-4 along various directions, as summarized in Fig.~\ref{dis-1}. Please refer to the supplement for a more detailed analysis~\cite{supplement}.

Here, we will use the $N=4$ case as an example to show that the exceptional behaviour in Hatano-Nelson model remains universal even in the presence of certain kinds of disorder, and the disordered Hamiltonians are homotopic to each other with respect to the tropicalization. The Hamiltonian, $H_4$, in the presence of a general form of scaling disorder reads

\begin{equation}
  H_{dis}=  \begin{pmatrix}
    0 & a\delta & \eta & \Delta \\
    (2+\delta)c & 0 & \delta b & 0 \\
    0 &  (2+\delta)d & 0 & \delta m \\
    0 & 0 &  (2+\delta)n & 0 \\
    \end{pmatrix},
\end{equation}

where $a, b, c, m$ and $n$ are arbitrary real numbers that introduce disorder in the asymmetric hopping terms. Such models are well-studied and can be experimentally realized in different physical settings~\cite{xiao2019anisotropic}. The form of the characteristic equation now changes, but remarkably, its tropicalization remains the same as for $H_4$.

\begin{equation}
\text{trop} \left( p(r,\lambda)(\omega) \right) = \text{min}(4\omega, 2\omega+1, \omega+1,1),
\end{equation}

for $\cos{\theta}, \sin{\theta}, \cos{\phi}, \sin{\phi} \neq 0$. The tropical polynomial remains invariant to the values of disorder scaling parameters suggesting the exceptional behaviour remains invariant, or is universal in the presence of disorder. Our framework makes this apparent through the tropicalization. A complementary view is to analyze the holonomy around the EPs. Consider varying some system parameters to form a loop in the parameter space while simultaneously tracing the evolution of the complex eigenmodes. If the loop encloses an EP-$N$, $N$ eigenmodes would undergo a cyclic permutation, which can be understood using holonomy matrices~\cite{mehri2008geometric,ryu2012analysis}. Whereas, if the loop marginally touches an EP-$N$, the projection of the eigenmode evolution forms $N$ petals in the complex energy plane, as shown in Fig.~\ref{dis-1}c. We used such a marginally touching loops to study the holonomy properties for $H_{dis}$. We find that in the presence of disorder, the eigenvalues get scaled, however their holonomy properties do not change, as shown in insets of Fig.~\ref{dis-1}e-g. As the tropicalization remains invariant, the set of disordered Hamiltonians are homotopically connected and the EPs are universal.

\section{Outlook}
Our work opens up several avenues for exploration. While we have formulated the tropical geometric framework for a single variable, we envisage that it should be possible to generalize this to several variables -- this will allow treating multiple perturbations on the same footing. It will be interesting to use our approach to classify the different non-Hermitian symmetry classes, and explore potential connections of tropical geometry to $K$ theory~\cite{gong2018topological}. Since our approach allows treating disorder in a natural way, it could be interesting to connect tropical geometry and random matrices, which have applications in many different fields of physics~\cite{guhr1998random}. We also expect our analytical approach to be practically useful for tuning to EPs and identifying conditions for NHSE in various experimental arenas. Finally, we note that, very recently, amoebas have been used to determine the generalized Brillouin zone for non-Hermitian systems~\cite{wang2022amoeba}.
In summary, we have introduced and developed a new framework to characterize EPs using tropical geometry. We have illustrated its implications using paradigmatic SSH and Hatano-Nelson models. Our work, bridging the fields of tropical geometry and non-Hermitian phenomena, is particularly timely given the surge of interest in non-Hermitian systems. We hope that our findings motivate further synergy between mathematics and non-Hermitian physics.

\subsection{Acknowledgments}
A.B. thanks the Prime Minister's Research Fellowship. R.J. thanks the Kishore Vaigyanik Protsahan Yojana fellowship. M.M was supported by a MATRICS grant from the Department of Science and Technology (DST) India. A.N. is supported by the startup grant of the Indian Institute of Science (SG/MHRD-19-0001) and by DST-SERB (project number SRG/2020/000153).
A part of this work was carried out during the program ``Combinatorial Algebraic Geometry: Tropical and Real" held at the International Centre for Theoretical Sciences, Bangalore (ICTS) in June-July, 2022. We thank ICTS for their warm hospitality. 
We thank Vamsi P. Pingali for bringing us together on this project. We thank G. Reddy and N. Aetukuri for feedback on the manuscript.
MM also thanks Indian Institute of Science, Bangalore and Waseda University, Tokyo for their hospitality during the visits in May, 2022 and January, 2023, respectively, in which a part of the work was carried out. 

\subsection{Author contributions}
A.B. and R.J. carried out the calculations in consultation with M.M. and A.N. All authors analyzed the results, developed the theory, and wrote the manuscript. 

\section{Methods}
\setlength{\parskip}{0\baselineskip}

\textbf{Fundamentals of tropical geometry.} Here we summarize some of the fundamentals of tropical geometry. The algebraic structure of tropical geometry is also known as the min-plus algebra. Many of the usual axioms of arithmetic remain valid in the tropical setting. For instance, addition and multiplication are commutative

\begin{equation}
    x \oplus y=y \oplus x, \quad x \odot y=y \odot x.
\end{equation}

Associative property also holds, as does the distributive law

\begin{equation}
    x \odot (y \oplus z)=x \odot y \oplus x \odot z.
\end{equation}

Both tropical operations have an identity element -- infinity for addition and zero for multiplication.

\begin{equation}
    x \oplus \infty = x, \quad  x \odot 0 = x.
\end{equation}

A distinct feature of tropical arithmetic is the absence of subtraction operation. On the other hand, tropical division is the classical subtraction. So, $(\mathbb{R}\cup \{\infty\},\oplus,\odot)$ satisfies all the ring axioms except for the existence of additive inverse -- such algebraic structures are termed \emph{semirings}.

\textbf{Newton polygon formalism.} We briefly discuss the Newton polygon formalism which is dual to amoebas. Let us start with the Puiseux series solution of the equation $f(x,y)=0$ in a suitable neighbourhood of the origin (in our case an EP ($\nu=0$)). Any polynomial $f(x,y)$ with the form

\begin{equation}
f(x,y)= \sum_{\eta,\zeta} a_{\eta \zeta} x^{\eta} y^{\zeta},
    \label{s1}
\end{equation}

admits a solution $y=t x^{\mu}$, where $t$ is a complex number and $\mu=p/q$ is a positive rational number. One can find a solution by substituting $y=t x^{\mu}$ in Eq.~\ref{s1}, to obtain

\begin{equation}
f(x,tx^{\mu}) = \sum_{\eta,\zeta} a_{\eta \zeta}x^{\eta+ \mu \zeta} t^{\zeta}
                    = x^{\xi} \sum_{\eta,\zeta} a_{\eta \zeta} t^{\zeta}.
\end{equation}

The above equation puts a constraint that $f(x,y)$ contains only monomials $x^{\eta} y^{\zeta}$ for which $\eta + \mu \zeta =\xi$, which is the essential feature of the Newton polygon. The geometric interpretation of Newton polygon is embedded in the following mapping. Each monomial $x^{\eta}y^{\zeta}$ maps to the pair $(\eta,\zeta)$ of natural numbers comprising a set of $\mathbb{N}^2$ lattice points with integer coordinates for non-zero coefficients $a_{\eta \zeta}\neq 0$. This set of lattice points forms the carrier $\Delta (f)$ of $f$, thus

\begin{equation}
    \Delta (f) = \{(\eta, \zeta) \in \mathbb{N}^2|a_{\eta \zeta}\neq 0\}.
\end{equation}

For a convergent power series $f(x,y)$ with a carrier $\Delta \left( f \right)$, one can define a convex hull from each point of the carrier $\Delta \left( f \right)$. The boundary of the convex hull, delineating a compact polygonal path, gives the Newton polygon of $f$. The steepest segment of the Newton polygon gives the lowest order term for the Puiseux series solution, thus defining the order of EP~\cite{brieskorn2012plane,jaiswal2021characterizing}. More concretely, the condition $\eta + \mu \zeta =\xi$ for all $(\eta,\zeta) \in \Delta (f) $ indicates that all points of $\Delta (f)$ lie on a line, with a slope $-\frac{1}{\mu}$, and the line meets the $\alpha-$axis at $\eta=\xi$.

\textbf{Amoeba and tropicalization.} We next present the connection between amoeba and tropicalization as used in the main text. The absolute value $|.|$ over the complex numbers satisfies the archimedean property~\cite[Chapter 9, page 313]{cohn2012basic}. Any field $F$ has an absolute value $|.|_t$ that is non-archimedean, i.e., does not satisfy the archimedean property: $|0_F|_t=0$ and $|c|_t=1$ for all $c \neq 0_F$ in $F$, where $0_F$ is the additive identity of $F$. This is usually called the \emph{trivial absolute value} on $F$. Otherwise, the non-archimedean absolute value is called \emph{non-trivial}.  Fields such as the rational numbers $\mathbb{Q}$, the field $\mathbb{C}((t))$ of formal Laurent power series in one variable (with complex coefficients) are naturally equipped with non-trivial (non-archimedean) absolute values. More explicitly, i. $|n|_p:=e^{-{\rm val}_p(n)}$ for any $n \in \mathbb{Q}$ where $p$ is a prime, ${\rm val}_p(n)$ is ${\rm ord}_p(r)-{\rm ord}_p(s)$ where $n=r/s$ such that $r, s \in \mathbb{Z}$, $s \neq 0$ and ${\rm ord}_p(i)$, for an integer $i$, is the largest power of $p$ that divides $i$, ii. $|\ell(z)|$ for a Laurent power series $\ell(z)$ is defined as $e^{-{\rm ord}(\ell(z))}$ where ${\rm ord}(\ell(z))$ is the least exponent of $z$ in the support of $\ell$. The rational number ${\rm ord}(\ell(z))$ is also called the \emph{valuation} of $\ell(z)$ and is denoted by ${\rm val}(\ell(z))$.

Suppose that $\mathbb{K}$ is an algebraically closed field (every polynomial of degree at least one in $\mathbb{K}[t]$ has a root) equipped with a non-trivial, non-archimedean absolute value $|.|_{\mathbb{K}}$.  Our primary example of such a field is the field  $\mathbb{C}\{\{t\}\}$ of Puiseux series is one variable. The notion of amoeba that we defined over the complex numbers can also be mimicked over $\mathbb{K}$ as follows. Suppose that $V \subseteq (\mathbb{K}^{\star})^n$ is the set of solutions, all of whose coordinates are non-zero, to a finite set of Laurent polynomials in $n$-variables with coefficients in $\mathbb{K}$. Let ${\rm Log}_{\mathbb{K}} :(\mathbb{K}^{\star})^n \rightarrow \mathbb{R}^n$ be the map ${{\rm Log}_{\mathbb{K}}}(s) = -\log |s|_{\mathbb{K}}$ (note that $s \neq 0_{\mathbb{K}}$) \footnote{Note that the negative sign in the definition of ${{\rm Log}_{\mathbb{K}}}$ is only for compatibility with the associated valuation and does not cause any essential change}. The tropicalization of $V$ is defined as the image of the map ${{\rm Log}_{\mathbb{K}}}$ restricted to $V$. Hence, the tropicalization of $V$ is a non-archimedean analogue of an amoeba.

\bibliography{references}

\end{document}